\documentclass[letterpaper, 10 pt, conference]{ieeeconf}

\usepackage{cite}

\usepackage{amsthm}
\usepackage{amssymb}
\usepackage{amsfonts,amsmath}
\usepackage{mathtools}
\DeclarePairedDelimiter\ceil{\lceil}{\rceil}
\DeclarePairedDelimiter\floor{\lfloor}{\rfloor}
\usepackage[font=footnotesize]{caption}
\usepackage{subcaption}
\usepackage{graphicx}
\usepackage{color, soul}
\usepackage{url}

\usepackage[usenames,x11names,table]{xcolor}
\usepackage{xspace} 
\usepackage{hyperref}
\usepackage{accents}
\usepackage{lipsum}

\usepackage[linesnumbered,vlined,ruled]{algorithm2e}
\usepackage{multirow} 
\usepackage{array} 

\newtheorem{prop}{Proposition}
\newtheorem{problem}{Problem}

\newtheorem{corollary}{Corollary}
\newtheorem{theorem}{Theorem}

\newtheorem{definition}{Definition}
\newtheorem{example}{Example}
\newtheorem*{excont}{Example \continuation}
\newcommand{\continuation}{??}
\newenvironment{continueexample}[1]
 {\renewcommand{\continuation}{\ref{#1}}\excont[\textit{Cont'd}]}
 {\endexcont}
 
 \newcommand{\xs}{\mathbf{x}}
 \newcommand{\us}{\mathbf{u}}

\mathtoolsset{showonlyrefs}
\title{\LARGE \bf Temporal Relaxation of Signal Temporal Logic Specifications for Resilient Control Synthesis}

\author{Ali Tevfik Buyukkocak and Derya Aksaray \thanks{A.T Buyukkocak is with the Department of Aerospace Engineering and Mechanics, University of Minnesota, Minneapolis, MN {\tt\small \{buyuk012@umn.edu\}}, D. Aksaray is with the Department of Electrical and Computer Engineering, Northeastern University, Boston, MA {\tt\small \{d.aksaray@northeastern.edu\}}. \newline \indent This work was partially supported by MnDRIVE Graduate Research Fellowship from the University of Minnesota.}}

\IEEEoverridecommandlockouts
\begin{document}
\bibliographystyle{IEEEtran}
\maketitle
\thispagestyle{empty}
\pagestyle{empty}

\begin{abstract} \label{abstract}
We introduce a metric that can quantify the temporal relaxation of Signal Temporal Logic (STL) specifications and facilitate resilient control synthesis in the face of infeasibilities. The proposed metric quantifies a cumulative notion of relaxation among the subtasks, and minimizing it yields to structural changes in the original STL specification by i) modifying time-intervals, ii) removing subtasks entirely if needed. To this end, we formulate an optimal control problem that extracts state and input sequences by minimally violating the temporal requirements while achieving the desired predicates. We encode this problem in the form of a computationally efficient mixed-integer program. We show some theoretical results on the properties of the new metric. Finally, we present a case study of a robot that minimally violates the time constraints of desired tasks in the face of an infeasibility.


\end{abstract}


\vspace{-5mm}
\textcolor{black}{\section{Introduction} \label{sec:intro}}
Motion planning and control of cyber-physical systems often require the satisfaction of complex tasks with time constraints.  One way of expressing such tasks is via temporal logics. In particular, Signal Temporal Logic (STL) \cite{maler} is an expressive specification language that can define properties 
with explicit spatial and time parameters. Different than other temporal logics, STL contains \emph{predicates} in the form of inequalities and is endowed with a metric called robustness degree that can quantify how well a signal satisfies an STL specification \cite{donze}. Such a metric facilitates optimization based control synthesis with STL specifications \cite{karaman,raman,pant2017smooth}.

There are two main STL robustness degree metrics, i.e., space and time robustness which capture the effect of shifting the signal in space and time on the satisfaction, respectively \cite{donze}. While most of the studies in the literature propose solution methods to optimize space robustness or a variant of it \cite{raman,pant2017smooth,akazaki2015time,rodionova2016temporal,lindemann2019average,haghighi2019control,mehdipour2019arithmetic,buyukkocak2021control}, there are also some works emphasizing the benefits of optimizing time robustness \cite{lin2020optimization,rodionova2021time,lindemann2022temporal}. Moreover, allowing for negative robustness and maximizing it may provide a notion of minimal violation of an STL specification.

In real-world applications, autonomous systems operate under various disturbances (e.g., internal, external, human-triggered), and there might be cases where the desired STL specification becomes infeasible. In such cases, a resilient control synthesis requires to search for trajectories resulting in minimal violations of the original specification. One common way of achieving this is by relaxing the spatial requirements of the specification (e.g., \cite{mehdipour2020specifying}), such as satisfying $x\geq4.7$ instead of $x\geq5$. On the other hand, relaxing the time bounds by strictly enforcing the thresholds over signal values is not investigated much in the STL literature, although some applications (e.g., manipulating objects in fixed regions) would benefit from it. Such a notion of temporal relaxation has been introduced for Time Window Temporal Logic in \cite{aksaray2016dynamic,vasile2017time}, where an automata-theoretic approach is proposed to shrink or extend the corresponding time intervals of the tasks when needed.

In this paper, we introduce a metric that can quantify the temporal relaxation of STL specifications and facilitate resilient control synthesis in the face of infeasibilities. The proposed metric quantifies a cumulative notion of time relaxation among the subtasks. We propose a mixed-integer encoding for the temporal relaxation metric and formulate a computationally efficient mixed-integer program to minimize it. We present some theoretical results on the properties of the new metric. Finally, we show the efficacy of the new metric in a case study with an autonomous robot and compare the resulting behavior with the one obtained by optimizing existing metrics. 

This work is closely related to  \cite{lin2020optimization}, \cite{rodionova2021time,lindemann2022temporal,ghosh2016diagnosis,aksaray2021resilient}. Time robustness in control synthesis is considered in \cite{lin2020optimization}, where the goal is to maximize the time robustness along with the space robustness for a limited family of STL specifications; however, the infeasibility of STL specifications is not discussed. The authors of \cite{rodionova2021time} introduce the mixed-integer encoding of the standard time robustness \cite{donze} by counting the consecutive satisfactions and violations. Similarly, time robustness under time shifts in the stochastic signals is assessed along with a risk measure of STL failure in \cite{lindemann2022temporal}. While maximizing time robustness may lead to satisfaction of the STL specification even when the signal is shifted along the time (e.g., delay in the signal), it may not necessarily result in minimal temporal relaxation of the STL specifications in the presence of violations. This is due to the fact that the quantity of time robustness metric is dominated by the critical violations (i.e., use of $\min$/$\max$ functions), which makes it hard to differentiate individual task relaxations. Furthermore, maximizing time robustness especially in the violation cases  \cite{lin2020optimization,rodionova2021time} requires a significant computational effort, which is also illustrated in our case study. Alternatively, repairing STL specifications is also addressed in \cite{ghosh2016diagnosis}, where an iterative algorithm is proposed to repair infeasible STL specifications by relaxing thresholds over signal values, penalizing relaxations, and determining possible interval modifications based on the relaxed thresholds. However, this work is not equivalent to minimizing temporal relaxation and its iterative nature demands more computational effort. Finally, a notion of temporal relaxation is investigated in \cite{aksaray2021resilient}, where an algorithm is proposed to plan trajectories that satisfy iteratively changing STL specification (i.e., shifting STL). In particular, the original STL specification is considered but its time intervals are from the current time instant to the end of the time interval which allows for achieving STL tasks later. However, this work is limited to safety and persistence behaviors. To the best of our knowledge, there is no framework in the literature that minimizes a temporal relaxation metric facilitating both shrinkage and expansion of the time intervals and minimally changing the structure of the STL specification by partially removing subtasks in a computationally efficient manner.

\textit{Notation:}
The set of non-negative (positive) real and integer numbers are denoted by $\mathbb{R}_{\ge0}$ ($\mathbb{R}^{+}$) and $\mathbb{Z}_{\ge0}$ ($\mathbb{Z}^{+}$), respectively. The set of $n$-dimensional real valued vectors is denoted by $\mathbb{R}^n$. Floor and ceiling functions are designated by $\floor*{\cdot}$ and $\ceil*{\cdot}$, respectively. The set cardinality is represented by $\vert\cdot\vert$. Minkowski sum of two sets $A,B\subset\mathbb{Z}$ is denoted by $\oplus$ where $A\oplus B=\{a+b\ |\ a\in A, b\in B\}$. With a slight abuse of notation, we also allow $a\oplus B=\{a+b\ |\ b\in B\}$. 

\textcolor{black}{\section{Signal Temporal Logic} \label{sec:STL}}
Signal Temporal Logic (STL) \cite{maler} can express rich properties of time series. In this paper, we specify desired system behaviors with the following expressive STL fragment\footnote{As any STL specification can be represented in negation-free form \cite{ouaknine2008some}, we do not present \textit{negation} operator throughout the paper. Moreover, the Until operator of STL is intentionally not included in this paper as there are multiple ways of its temporal relaxation.}:
\begin{equation}
\small
\label{eq:STL_fragment}
\begin{split}
\Phi &::= \phi\  |\  \Phi_1\wedge\Phi_2\ |\ {\Phi_1\vee\Phi_2}\ |\ F_{[a,b]}\Phi\ |\ G_{[a,b]}\Phi,\\
\phi &::=  F_{[a,b]}\varphi\ |\ G_{[a,b]}\varphi, \\
\varphi &::= \mu\ |\ \varphi_1\wedge\varphi_2 \ |\ \varphi_1\vee\varphi_2,
\end{split}   
\end{equation}


\noindent where $\Phi, \phi, \varphi$ are STL specifications and subtasks; $\mu$ is a predicate in an inequality form such as $\mu=p(\xs_t)\geq\, 0$ with the value of a discrete-time signal $\xs:\mathbb{Z}_{\ge0} \to\mathbb{R}^n$ at time $t$ and a function $p:\mathbb{R}^n \to \mathbb{R}$. 
While $(\xs,t)\models F_{[a,b]}\Phi$ implies that $\Phi$ holds at some time instant within $t\oplus[a,b]$, $(\xs,t)\models G_{[a,b]}\Phi$ requires $\Phi$ to hold at all time instants within the same time interval, where $a,b\in\mathbb{Z}_{\ge0}$ are finite
time bounds with $b\!>\! a$. It is worth noting that the major restriction of the syntax in \eqref{eq:STL_fragment} is not allowing for the conjunction or disjunction of the temporal operators with bare predicates such as $\Phi=F_{[a,b]}(\mu_1\wedge F_{[c,d]}\mu_2)$.


STL is endowed with real-valued functions, robustness degree, that are used to quantify the satisfaction of an STL specification $\Phi$ with respect to a signal $(\mathbf{x},t)$. While positive robustness degree indicates the satisfaction of $\Phi$, negative one represents a violation. The space and time robustness functions can be formally defined as below:

\begin{definition} (STL Space Robustness \cite{donze}) Space robustness is a real-valued function $\rho(\cdot) \in \mathbb{R}$ that is used to quantify how much the signal $\xs$ can be shifted in space such that the specification is still satisfied. It is defined recursively as:  
\begin{equation}
\small
\setlength{\jot}{.8pt}
\begin{split}
\rho(\mathbf{x},\mu,t ) &:=p(\mathbf{x}_t), \\
\rho(\mathbf{x},\Phi_{1}\wedge\Phi_{2},t)& :=\min\big(\rho(\mathbf{x},\Phi_{1},t),\rho (\mathbf{x},\Phi_{2},t) \big), \\
\rho(\mathbf{x},\Phi_{1}\vee\Phi_{2},t)& :=\max \big(\rho(\mathbf{x},\Phi_{1},t ) ,\rho  (\mathbf{x}, \Phi _{2},t)\big), \\
\rho(\mathbf{x},F_{[a,b]}\Phi,t)&:=\mathop{\max }_{\mathop{t}^{'}\in t\oplus[a,b]}\rho(\mathbf{x}, \Phi ,t'),\\
\rho(\mathbf{x},G_{[a,b]}\Phi,t)&:=\mathop{\min }_{\mathop{t}^{'}\in t\oplus[a,b]}\rho(\mathbf{x}, \Phi ,t').
\end{split}
\label{eq:space_robustness}
\end{equation}
\end{definition}

\begin{definition} (STL Time Robustness \cite{donze}) Right\! $(\!\texttt{+})$ and left\! $(\!\texttt{-})$
time robustness of an STL specification quantify how much the signal $\xs$ can
be shifted either right or left in time such that the specification is still satisfied. The time robustness of $\xs$ with respect to a predicate $\mu$ can be computed as:
\vspace{-1mm}
\begin{equation}\label{eq:time_robustness}
    \vspace{-4mm}
\small
\setlength{\jot}{.1pt}
\begin{split}
    \theta^+(\xs,\mu,t)&:=\mathcal{X}(\mu,t)\cdot\max\Big\{ d\geq0\; \textit{s.t.}\; \forall t'\in [t,t+ d],\\ &\hspace{43mm}\mathcal{X}(\mu,t')=\mathcal{X}(\mu,t)\Big\},\\
    \theta^-(\xs,\mu,t)&:=\mathcal{X}(\mu,t)\cdot\max\Big\{ d\geq0\; \textit{s.t.}\; \forall t'\in [t- d,t],\\ &\hspace{43mm}\mathcal{X}(\mu,t')=\mathcal{X}(\mu,t)\Big\},
\end{split}
\end{equation}
where the characteristic function $\mathcal{X}(\mu,t)$ is
\begin{equation}
\small
\begin{matrix}
\mathcal{X}(\mu ,t):=\; \left\{\; \begin{matrix*}[l]
-1, & p(\xs_t)<0, \\
+1, & p(\xs_t)\geq 0.  \end{matrix*}\right.
\end{matrix}
\label{eq:char_eq}    
\end{equation}

\end{definition}

\noindent After computing the right/left time robustness of a predicate, the rules in \eqref{eq:space_robustness} can be used to quantify the overall time robustness of a signal with respect to an STL specification. There also exists a combined notion called space-time robustness, where the inequalities in the characteristic equation are defined based on a desired space robustness threshold \cite{donze}. 

\begin{example} \label{example_1}
Consider two signals $\xs'$ and $\xs''$ starting from $t=0$ and illustrated in Fig.~\ref{fig:example1}. Suppose that the specification is $\Phi_1=G_{[15,60]}\xs\geq h_1$ meaning that the signal has to satisfy $\xs_t\geq h_1$ for all $t \in [15,60]$. While both signals violate $\Phi_1$, $\xs'$ (blue) stays in the desired region longer than $\xs''$ (red). However, the (right) time robustness notion cannot differentiate this due to using $\min$ function in its computation as in \eqref{eq:space_robustness}. Specifically, $\theta^+(\xs',\Phi_1,0) = \theta^+(\xs'',\Phi_1,0) = \min_{t \in [15,60]} \theta^+(\xs,\xs_t\geq h_1,t)$, hence both signals have the same time robustness that is negative due to violation, and its length is illustrated by the purple bar in Fig.~\ref{fig:example1}.

Now, suppose the specification is updated as $\Phi_2=\Phi_1\wedge F_{[75,120]}\xs\leq h_2$
meaning that the signal needs to satisfy $\xs_t\leq h_2$ for some time $t\in [75,120]$ in addition to satisfying $\Phi_1$. Despite the blue signal's satisfaction of the second subtask, the violation of $\Phi_1$ dominates the computation of time robustness of $\Phi_2$. Furthermore, the signal $\xs''$ slightly violates the second subtask, but the initial violation still dominates. As a result, the violation in the purple area indifferently determines the time robustness score of both signals and specifications with $\theta^+(\xs',\Phi_1,0)=\theta^+(\xs'',\Phi_1,0)=\theta^+(\xs',\Phi_2,0)=\theta^+(\xs'',\Phi_2,0)$. Similar examples are straightforward to illustrate for the left time robustness as well.
\end{example}

\begin{figure}[htb!]
    \centering
    \includegraphics[trim=80 7 100 25,clip, width=\linewidth]{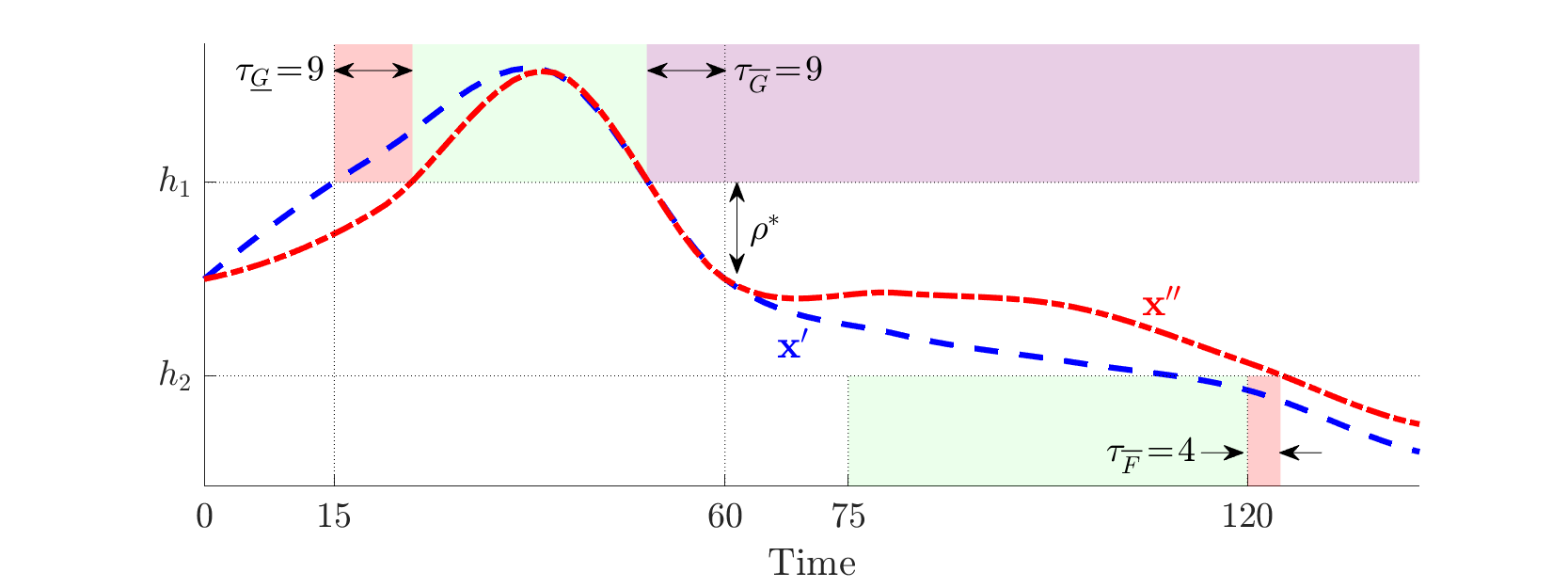}
    \caption{Two signals that violate the given STL specifications by different amounts. Green regions are the targets the system must reside in to achieve the specifications. The areas of violation is denoted by the color of the respective signal, while the purple region depicts violation by both signals.}
    \label{fig:example1}
\end{figure}

Note that the standard space robustness in \eqref{eq:space_robustness} may be insufficient to differentiate signal behaviors as well. For instance, $\rho(\xs',\Phi_1,0)=\rho(\xs'',\Phi_1,0)=-\rho^*$. However, this issue is addressed by several studies via modified space robustness metrics and extended STL syntax (e.g., \cite{akazaki2015time,rodionova2016temporal,lindemann2019average,haghighi2019control,mehdipour2019arithmetic,buyukkocak2021control}). Nonetheless, such modified metrics do not assess satisfaction beyond the original time intervals, which motivate us to introduce a temporal relaxation metric for STL specifications.

One may also consider to use average time robustness values over the intervals instead of $\min$/$\max$ functions similar to existing space robustness modifications. However, optimizing average time robustness may not result in desired relaxed notions. For instance, let $\Phi_3=G_{[75,120]}\xs\geq h_2$. As depicted in Fig.~\ref{fig:example1}, $\xs'$ leaves the $\xs\geq h_2$ region before the deadline of $t=120$. This yields the time robustness values between the violation instant and $t=120$ to become very small (negative) since there is no satisfaction instant remained until the end of the mission. Hence, arbitrarily large mission horizons yield violated time robustness values to be arbitrarily small. In this regard, averaging them with the satisfactory positive small values may not provide the desired understanding of the signal behavior.

{\section{Minimal Temporal Relaxation of STL Specifications} \label{sec:minimal_temp_rlx_STL}}
The spatial relaxation of STL specifications is generally possible by allowing the system to minimally violate its spatial requirements (i.e., maximizing the space robustness while letting it be negative). For example, if the specification is entering a desired region, the spatial relaxation will be approaching the region as close as possible given the time interval.  
However, spatial relaxation may not be feasible for some missions, e.g., manipulating an object in a fixed region.
Therefore, extending or shrinking the time intervals to allow for
late or early satisfaction can be practical in some scenarios.
Nonetheless, the temporal relaxation of STL is not broadly investigated, and the standard time robustness cannot differentiate partial satisfactions as illustrated in Example \ref{example_1}.



\subsection{Temporal Relaxation of STL}\label{sec:STL_temp_relax}
Suppose that a subtask $\phi$ in the STL specification $\Phi$ (as per the syntax in \eqref{eq:STL_fragment}) is defined over a time interval of $I=[a,b]$. For every such subtask, we have the temporally relaxed versions $\phi^\tau= F_{[a-\tau_{\underline{F}},b+\tau_{\overline{F}}]}\varphi$ and $\phi^\tau= G_{[a+\tau_{\underline{G}},b-\tau_{\overline{G}}]}\varphi$ where $\tau_{\underline{F}},\tau_{\overline{F}},\tau_{\underline{G}},\tau_{\overline{G}}\in\mathbb{Z}_{\geq0}$ are temporal relaxation parameters. That is, the system may \textit{eventually} satisfy $\varphi$ within a longer time interval and \textit{always} satisfy it within a shorter interval, both of which imply the relaxation of the original requirement. Moreover, when multiple temporal operators are nested, the temporal relaxation is enabled for the innermost ones. For example, if the original specification is $\Phi= F_{[c,d]}G_{[a,b]}\varphi$, then the relaxed version of it is considered as $\Phi^\tau= F_{[c,d]}G_{[a+\tau_{\underline{G}} ,b-\tau_{\overline{G}}]}\varphi$.


Our approach will result in the satisfaction of the predicates and/or their combinations only within the original \textit{globally} interval $I_G$ (potentially within an interval shorter than the original) and remove the subtask otherwise. On the other hand, the interval of the \textit{finally} operator $I_F$ is extended, and when not bounded, this extension may yield excessive relaxations, e.g, satisfying $\mu$ at $t=100$ where the specification is $\phi=F_{[5,10]}\mu$.  
For this reason, our formulation allows the user to specify acceptable bounds to relax the original interval $I=[a,b]$. Accordingly, the new interval bounding the maximal relaxation is defined as $\overline{I}_F=(a-\ceil*{\gamma_F\!\cdot\! \vert I_F\vert},b+\ceil*{\gamma_F\!\cdot\! \vert I_F\vert})$ where $\gamma_F>0$ is a user-defined tolerance parameter. 

A similar tolerance $0\!<\!\gamma_G\!\leq\!1$ can be defined for the \textit{globally} as well where the relaxation is allowed over $\overline{I}_G\!=\![a,b]\setminus\Big(\big[a,a\!+\!\floor*{\gamma_G\cdot\frac{\vert I_G\vert}{2}}\big)\cup\big(b\!-\!\floor*{\gamma_G\cdot\frac{\vert I_G\vert}{2}},b\big]\Big)$. Note that, $\overline{I}_F$ and $\overline{I}_G$ are not the relaxed intervals, but they exclusively bound the allowable relaxation according to user preferences. 




Next, we define the temporal relaxation metrics that we want to keep minimum. 
The temporal relaxation metric simply comprises the amount of violation with respect to the maximum allowable relaxation. We define two normalized metrics below to measure the temporal relaxation among the subtasks with \textit{finally} and \textit{globally} operators, respectively.

\begin{definition}
(Temporal Relaxation Metric) Given an STL specification with the syntax in \eqref{eq:STL_fragment}, the temporal relaxation metrics for the \textit{finally} and \textit{globally} operators are defined as:

\noindent\textbf{Finally} 
\begin{equation}\label{eq:tau_F}\small
        \vspace{-4mm}
\setlength{\jot}{.95pt}
\begin{matrix}
    \tau(\xs,F_{[a,b]} \varphi,t):=\left\{ \begin{split}
    &\frac{\max(\tau_{\underline{F}},\tau_{\overline{F}})}{\gamma_F\!\cdot\!\vert I_F\vert},\hspace{1mm} \textit{if}\;\; \exists t'\in t\oplus[a-\tau_{\underline{F}},b+\tau_{\overline{F}}]\\ 
    &\hspace{1.5mm}\textit{s.t.}\ \gamma_F\!\cdot\!\vert I_F\vert\geq\tau_{\underline{F}},\tau_{\overline{F}}\geq 0\ \textit{and}\ (\xs,t')\models\varphi,\\
&\hspace{8mm}1,\hspace{34.5mm} \textit{otherwise}.
    \end{split}\right.
    \end{matrix}
\end{equation}
\textbf{Globally} 
\begin{equation}\label{eq:tau_G}\small
\setlength{\jot}{.95pt}
\begin{matrix}
    \tau(\xs,G_{[a,b]} \varphi,t):=\left\{ \begin{split}&\frac{\tau_{\underline{G}}+\tau_{\overline{G}}}{\gamma_G\!\cdot\!\vert I_G\vert},\hspace{6.5mm}\textit{if}\;\; \forall t'\in t\oplus[a+\tau_{\underline{G}},b-\tau_{\overline{G}}]\\
   &\hspace{.5mm}\textit{s.t.}\; \gamma_G\!\cdot\!\frac{\vert I_G\vert}{2}\geq\tau_{\underline{G}},\tau_{\overline{G}}\geq 0\ \textit{and}\   (\xs,t')\models\varphi, \\
    &\hspace{5mm}1,\hspace{37mm} \textit{otherwise}.
    \end{split}\right.
    \end{matrix}
\end{equation}
\end{definition}
\vspace{-4mm}
\begin{prop}\label{prop_1}
Suppose that $\tau_{\underline{F}},\tau_{\overline{F}}\in\mathbb{Z}_{\geq0}$. Let $\phi\!=\!F_{[a,b]}\, \varphi$ be a finally subtask endowed with the temporal relaxation metric $\tau(\xs,F_{[a,b]} \varphi,t)$ defined according to \eqref{eq:tau_F}. When $\tau(\xs,F_{[a,b]} \varphi,t)$ is minimized, the resulting relaxed interval $[a-\tau_{\underline{F}},b+\tau_{\overline{F}}]$ can have at least one of $\tau_{\underline{F}}$ and $\tau_{\overline{F}}$ equal to zero.
\end{prop}
\vspace{-3mm}
\begin{proof}
  Without loss of generality, let $t=0$ and the required single satisfaction happen at $b+\tau_{\overline{F}}$. Suppose that the minimally relaxed interval can only be obtained as $[a-\tau_{\underline{F}},b+\tau_{\overline{F}}]$ with $\tau_{\underline{F}}\not=0$ and $\tau_{\overline{F}}\not=0$. Now, consider another interval $[a,b+\tau_{\overline{F}}]$ with $\tau_{\underline{F}}=0$. 
  Since $(\xs,t)\models F_{[a,b+\tau_{\overline{F}}]}\varphi \Longrightarrow (\xs,t)\models F_{[a-\tau_{\underline{F}},b+\tau_{\overline{F}}]}\varphi$ via the STL semantics in \eqref{eq:space_robustness}, $[a-\tau_{\underline{F}} , b+\tau_{\overline{F}}]$
with $\tau_{\underline{F}}\not=0$ and $\tau_{\overline{F}}\not=0$ cannot be the minimally relaxed interval. Hence, at least one of $\tau_{\underline{F}}$ and $\tau_{\overline{F}}$ can be zero when the interval is minimally relaxed. 
\end{proof}

For instance, consider $\Phi=F_{[5,10]}\varphi$, and the closest satisfaction is $(\xs,15)\models\varphi$. The relaxed specification would be $\Phi^\tau=F_{[5,15]}\varphi$. Alternatively, if $(\xs,3)\models\varphi$ is the closest satisfaction instant, then the relaxed specification would be $\Phi^\tau=F_{[3,10]}\varphi$. When $\tau_{\underline{F}}=\tau_{\overline{F}}=0$, then the original specification is satisfied, i.e., $(\xs,t)\models F_{[a,b]}\varphi$, without any temporal relaxation, i.e., $\tau(\xs,F_{[a,b]} \varphi,t)=0$.

Note that the temporal relaxation metrics in \eqref{eq:tau_F} and \eqref{eq:tau_G} are normalized and take values within $[0,1]$. That is, when no relaxation is needed and the original specification is achieved, the relaxation metric is $0$ for each operator, and when the allowable relaxation amount is exceeded or the subtask is completely failed, the metric becomes $1$ and the subtask is removed. Hence, the aim will be minimizing the temporal relaxation during the satisfaction of an STL specification $\Phi$ in the face of infeasibilities.

Starting from $\tau(\xs,F_{[a,b]} \varphi, t)$ and $\tau(\xs,G_{[a,b]} \varphi, t)$ in \eqref{eq:tau_F} and \eqref{eq:tau_G}, we recursively define a general temporal relaxation metric $\tau(\xs,\Phi,t)$ for the STL specification $\Phi$ as follows:
\begin{subequations}
\setlength{\jot}{.95pt}
\begin{align}
\tau(\xs,\Phi_{1}\wedge\Phi_{2}, t)& :=\frac{\tau(\xs,\Phi_{1}, t)\ +\ \tau(\xs,\Phi_{2}, t)}{2},\label{eq:temp_relax_metric_a}\noeqref{eq:temp_relax_metric_a} \\
\tau(\xs,\Phi_{1}\vee\Phi_{2}, t)& :=\min\big(\tau(\xs,\Phi_{1}, t), \tau(\xs,\Phi_{2},t)\big),\label{eq:temp_relax_metric_b}\noeqref{eq:temp_relax_metric_b} \\
\tau(\xs,G_{[a,b]}\Phi,t)& :=\mathop{\max }_{\mathop{t}'\in t\oplus[a,b]}\tau(\xs,\Phi,t'),\label{eq:temp_relax_metric_c}\noeqref{eq:temp_relax_metric_c}\\
\tau(\xs,F_{[a,b]}\Phi,t)& :=\mathop{\min }_{\mathop{t}'\in t\oplus[a,b]}\tau(\xs,\Phi,t').\label{eq:temp_relax_metric_d}\noeqref{eq:temp_relax_metric_d}
\end{align}
\label{eq:temp_relax_metric}
\end{subequations}
By definition, the temporal relaxation is allowed for the innermost temporal STL operators via \eqref{eq:tau_F} and \eqref{eq:tau_G} when multiple of them are nested. In other words, when $\Phi=F_{[a,b]}\varphi$ or $\Phi=G_{[a,b]}\varphi$, i.e., $\Phi=\phi$ as per \eqref{eq:STL_fragment}, we use temporal relaxation metrics in \eqref{eq:tau_F} and \eqref{eq:tau_G}, respectively. However, when there are multiple nested temporal operators, e.g., $\Phi=G_{[a,b]}F_{[c,d]}\varphi$, we use the definitions in \eqref{eq:temp_relax_metric_c} and \eqref{eq:temp_relax_metric_d} for the outer ones together with \eqref{eq:tau_F} and \eqref{eq:tau_G} inside them.

\begin{corollary}\label{corollary_1}
For a given STL specification $\Phi$ with the syntax in \eqref{eq:STL_fragment}, the overall relaxation metric defined according to \eqref{eq:temp_relax_metric} is bounded such that $0\leq\tau(\xs,\Phi,t)\leq1$, $\forall t$.
\end{corollary}
\vspace{-2mm}
\begin{proof}
This follows directly from the semantics of temporal relaxation metric in \eqref{eq:temp_relax_metric} together with \eqref{eq:tau_F} and \eqref{eq:tau_G}.
\end{proof}

It is worth noting that unlike the standard STL quantifying metrics (e.g., \cite{maler}), we use $\min$ for the disjunction or \textit{finally} (in \eqref{eq:temp_relax_metric_b} and \eqref{eq:temp_relax_metric_d}), and $\max$ for the \textit{globally} (in \eqref{eq:temp_relax_metric_c}). This is because the lower temporal relaxation metric implies the closer satisfaction to the original specification. Moreover, the averaging in \eqref{eq:temp_relax_metric_a} enables us to measure collective performance that is not dominated by critical values.

\begin{prop}\label{prop_2}
For a given STL specification $\Phi$ with the syntax in \eqref{eq:STL_fragment}, the temporal relaxation metric defined according to \eqref{eq:temp_relax_metric} is sound in the sense that $\tau(\xs,\Phi,t)\!=\!0\!\Longrightarrow\!(\xs,t)\!\models\!\Phi$\footnote{The reverse $\tau(\xs,\Phi,t)=0\Longleftarrow(\xs,t)\models\Phi$ may not always be true as the satisfaction over the original time intervals implies satisfaction over arbitrarily relaxed time intervals as well with nonzero relaxation parameters according to \eqref{eq:tau_F} and \eqref{eq:tau_G}. 
}.
\end{prop}
\begin{proof} 
By recursive definition in \eqref{eq:temp_relax_metric} together with \eqref{eq:tau_F} and \eqref{eq:tau_G}, no temporal relaxation, i.e., $\tau(\xs,\Phi,t)\!=\!0$, implies that all subtasks in $\Phi$ expressed according to the syntax in \eqref{eq:STL_fragment} have either $\tau(\xs,F_{[a,b]}\varphi,t')\!=\!0$ or $\tau(\xs,G_{[a,b]}\varphi,t')\!=\!0$ with $t'\geq t$. Meaning, for all relaxation parameters in $\Phi$, we have $\tau_{\underline{F}}\!=\!\tau_{\overline{F}}\!=\!\tau_{\underline{G}}\!=\!\tau_{\overline{G}}=0$ as per \eqref{eq:tau_F} and \eqref{eq:tau_G} which, by definition, yields satisfaction within the original interval $[a,b]$ for all subtasks. Hence, it follows by construction that the overall specification is satisfied on time, $(\xs,t)\!\models\!\Phi$.
\end{proof}
\vspace{-2mm}

Proposition \ref{prop_2} along with Corollary \ref{corollary_1} can be interpreted as follows: in the best case, the given STL specification is satisfied within the original time intervals with $\tau(\xs,\Phi,t)\!=\!0$. Therefore, an extra performance such as achieving tasks for longer time than required is not demanded unlike the optimization of other quantifying metrics of STL. However, when it is not feasible to achieve $\Phi$, the minimization of $\tau(\xs,\Phi,t)$ leads to the satisfaction of spatial requirements specified in $\Phi$ with a minimal temporal relaxation. Furthermore, this minimal relaxation may potentially yield structural changes on the $\Phi$ such as compromising highly demanding subtasks which may jeopardize the success of the others.

\begin{continueexample}{example_1} Consider again the two signals $\xs'$ (blue) and $\xs''$ (red) in  Fig.~\ref{fig:example1}. Recall that both signals have the same time robustness value even for different STL specifications $\Phi_1=G_{[15,60]}\xs\geq h_1$ and $\Phi_2=\Phi_1\wedge F_{[75,120]}\xs\leq h_2$. In this regard, the time robustness metric may not capture the signal behaviors completely. We now examine the performance of the new temporal relaxation metric for the same scenario. Let $I_1=[15,60]$ and $I_2=[75,120]$ with $\vert I_1 \vert=\vert I_2\vert = 46$. Assume that the tolerance parameters are defined as $\gamma_F=1$ and $\gamma_G=1$, and note the temporal relaxation parameters $\tau_{\overline{G}}=\tau_{\underline{G}}=9$ and $\tau_{\overline{F}}=4$ from Fig.~\ref{fig:example1}. As both signals violate $\Phi_1$, the violation amounts can be calculated using the proposed temporal relaxation metric as $\tau(\xs',\Phi_1,0)={9}/{46}$ and $\tau(\xs'',\Phi_1,0)={(9+9)}/{46}=18/46$ implying the temporal relaxation made for $\xs'$ is lower than that for $\xs''$. Similarly, when quantifying the violations for $\Phi_2$, the ones for $\Phi_1$ should neither dominate the computation nor be completely useless, but contribute to the calculation of the cumulative relaxation. Since $\xs'$ completely satisfies the \textit{finally} subtask, the only violation comes from $\Phi_1$ but averaged due to total number of two subtasks as $\tau(\xs',\Phi_2,0)=\frac{9}{2\cdot 46}$. On the other hand, $\xs''$ violates both subtasks leading to $\tau(\xs'',\Phi_2,0)=\big(\frac{18}{46}+\frac{4}{46}\big)/2=\frac{11}{46}$. These temporal relaxations for two signals yield the following relaxed specifications: $\Phi^{\tau'}_2=G_{[15,60-\tau_{\underline{G}}]}\xs\geq h_1\wedge F_{[75,120]}\xs\leq h_2$ for $\xs'$ requiring less relaxation than that of $\Phi^{\tau''}_2=G_{[15+\tau_{\underline{G}},60-\tau_{\overline{G}}]}\xs\geq h_1\wedge F_{[75,120+\tau_{\overline{F}}]}\xs\leq h_2$, the relaxed specification for $\xs''$.
\end{continueexample} 

\begin{definition}\label{def:similarity}
(Time Interval Similarity\footnote{Note that this similarity measure is a form of Jaccard similarity coefficient when one interval is contained in the other.}) The similarity between two time intervals $I_1$ and $I_2$ is measured via 
\begin{equation}
\small
    \mathcal{S}(I_1,I_2):=\frac{\vert I_1\cap I_2\vert}{\max\big(\vert I_1\vert,\vert I_2\vert\big)}.
\end{equation}
\end{definition}
\begin{prop}\label{prop_3}
Let $\phi$ be an STL subtask with a single temporal operator over an original interval of $I=[a,b]$ as in \eqref{eq:STL_fragment}, e.g., $\phi=F_{[a,b]}\, \varphi$ or $\phi=G_{[a,b]}\, \varphi$. Suppose that two feasible minimal relaxations of $\phi$, $\tau'(\xs,\phi,t)$, $\tau''(\xs,\phi,t)\in[0,1]$, are obtained over the relaxed intervals of $I'$ and $I''$, respectively. If $\tau'(\xs,\phi,t)\leq \tau''(\xs,\phi,t)$ under the same user tolerances $\gamma_F$ and $\gamma_G$, then $\mathcal{S}(I',I)\geq\mathcal{S}(I'',I)$.
\end{prop}
\begin{proof}
\textbf{\textit{Finally Case:}}  \textcolor{black}{As shown in Prop. \ref{prop_1}, minimal temporal relaxation of finally operator enables one of the relaxation parameters $\tau_{\underline{F}}$ or $\tau_{\overline{F}}$ to be equal to zero.} Without loss of generality, suppose that the relaxations take place at the right hand side (future) of the original intervals, i.e., $\tau'_{\underline{F}}=\tau''_{\underline{F}}=0$. Then the condition in the premise implies $\tau'_{\overline{F}}\leq\tau''_{\overline{F}}$ according to \eqref{eq:tau_F}. It follows from the relaxed intervals $I'\!=\![a,b+\tau'_{\overline{F}}]$ and $I''\!=\![a,b+\tau''_{\overline{F}}]$ that  $\mathcal{S}(I'\!,I)\!\geq\!\mathcal{S}(I''\!,I)$ as per Def. \ref{def:similarity}.

\noindent\textbf{\textit{Globally Case:}} Suppose that the relaxed intervals are $I'=[a+\tau'_{\underline{G}},b-\tau'_{\overline{G}}]$ and $I''=[a+\tau''_{\underline{G}},b-\tau''_{\overline{G}}]$. Then the condition in the premise implies $\tau'_{\underline{G}}+\tau'_{\overline{G}}\leq\tau''_{\underline{G}}+\tau''_{\overline{G}}$ according to \eqref{eq:tau_G}; thereby, $\mathcal{S}(I',I)\geq\mathcal{S}(I'',I)$ as per Def. \ref{def:similarity}.
\end{proof}
\vspace{-1mm}
Hence, keeping the temporal relaxation minimum is desired to achieve a relaxed specification as close as possible to the original one. 

\begin{problem} \label{problem_1}
Given an initial state $\xs_0$, an STL specification $\Phi$, some user-defined relaxation tolerances $\gamma_F$ and $\gamma_G$, the minimal temporal relaxation problem for a dynamical system can be formulated as an optimization problem as follows:
\begin{equation}\small
\begin{split}
\us^*=arg\,min \;& \tau(\xs,\Phi,0) \\
s.t.\;\;\;\;&\xs^+_t=f(\xs_t,\us_t),\\
&\xs_t\in\mathcal{X},\ \us_t\in\mathcal{U},\ \forall t,
\end{split}
\label{eq:optimization}
\end{equation}

\noindent where $f:\mathbb{R}^n\times\mathbb{R}^m\to\mathbb{R}^n$ denotes the dynamics of the system; $\mathcal{X}$ and $\mathcal{U}$ are the admissible state and control sets, respectively.
\end{problem}

If $\Phi$ can be satisfied by a trajectory starting from the given initial state $\xs_0$, then solving Problem \ref{problem_1} generates a feasible trajectory. Otherwise, solving it results in a trajectory that minimally relaxes the time bounds of $\Phi$ under the allowable relaxations provided by the user. Such a trajectory is different than the one that can be found by maximizing space robustness (which does not modify the original time bounds) and time robustness. 
\textcolor{black}{\subsection{Temporal Relaxation vs. Time Robustness}} \label{sec:temp_rlx_vs_TR}

Given an initial state $\xs_0$, if $\Phi$ cannot be satisfied by a trajectory starting from $\xs_0$, the trajectories found by minimizing temporal relaxation and maximizing time robustness are not the same. In the literature, either right or left time robustness is maximized and addressing both past and future relaxations at the same time is missing. Even if a unified approach can be proposed to account for this, there exists a major difference regarding the computation of time robustness, which uses $\min$/$\max$ functions at each level. This leads to the major violation cases to dominate the value of the overall time robustness and cannot differentiate the satisfaction/violation of the other subtasks (as illustrated in Example \ref{example_1}). Furthermore, the time robustness becomes inconclusive when it is equal to zero. This may potentially yield problems in
differentiating between a near miss and a tight satisfaction. For instance, let a specification be given as $\phi=F_{[5,10]}\mu$ with two scenarios: either $(\xs,10)\models\mu$ (tight satisfaction) or $(\xs,11)\models\mu$ (near miss). As the (right) time robustness is calculated by $\theta^+(\xs,\phi,0)\!=\!\max_{t\in{[5,10]}}\theta^+(\xs,\mu,t)\!=\!\theta^+(\xs,\mu,10)\!=\!0$ for both cases, we have the same time robustness score for them. Hence, satisfying $\phi$ within the specified time interval but at the very last step or violating it by one time step becomes indifferent. Hence, determining the need of relaxing each subtask is not possible by using the notion of time robustness.

Moreover, maximizing time robustness in the presence of a violation leads to checking the overall mission horizon to ensure the satisfaction of the corresponding subtask (or predicate) at least once so that a finite negative time robustness can be obtained. However, this can cause the mandatory satisfaction instant to be arbitrarily far from the original interval which may not be desirable for some missions. On the other hand, the proposed formulation can accommodate user input for maximum relaxations and minimize temporal relaxation under the allowable relaxation tolerances. If a subtask requires relaxation more than the allowable relaxation, then that subtask is removed from the specification. Note that a similar subtask removal by maximizing time robustness might be possible. However, that becomes a more computationally intensive process as it requires to iteratively compute trajectories and their corresponding time robustness values and decide to remove a subtask if the resulting time robustness is smaller than a user-defined threshold. On the other hand, the proposed approach of minimizing temporal relaxation does not require iterative trajectory computation and solves the problem in single shot while still incorporating the user preferences.

\vspace{-2mm}
\textcolor{black}{\section{Solution Approach} \label{sec:soln_approach}}
Mixed-integer encoding of STL constraints is a common approach in control synthesis under an STL specification \cite{raman}. If the dynamics $f:\mathbb{R}^n\times\mathbb{R}^m\to\mathbb{R}^n$ is linear (e.g., $\xs_{t+1}=A\, \xs_t+B\, \us_t$ where $A\in\mathbb{R}^{n\times n}$ and $B\in\mathbb{R}^{n\times m}$ are the system matrices), the mixed-integer encoding of the STL specification $\Phi$ and the temporal relaxation metric $\tau(\xs,\Phi,0)$ renders the optimization problem in \eqref{eq:optimization} a mixed-integer linear program (MILP) as long as the predicate function $p(\xs_t)$ is linear as well.

The authors of \cite{rodionova2021time} use the framework in \cite{raman} to maximize right time robustness by counting the consecutive satisfactions and violations via MILP encoding. While we use a similar idea of counting constraints, there are substantial differences as the counting take place within particular time intervals and the counters are used to encode the temporal relaxation metric according to \eqref{eq:tau_F}, \eqref{eq:tau_G}, and \eqref{eq:temp_relax_metric}.

We start by encoding the satisfaction of predicates via the standard big-M approach using binary variables $z_t\in\{0,1\}$ as follows:
\vspace{-2mm}
\begin{equation}
\small
\begin{matrix}
\mu =p(\xs_t) \ge 0\;\;\Longleftrightarrow \left\{\quad \begin{matrix*}[l]
p(\xs_t)\ge M( z^{\mu}_t-1), \\
p(\xs_t)\le M z^{\mu}_t-\varepsilon, \end{matrix*}\right.
\end{matrix}
\label{eq:big_M}    
\end{equation}
\normalsize
\noindent
where $M, \varepsilon \in \mathbb{R}^{+}$  are sufficiently large and small numbers, respectively. Note that $\varphi$ includes the predicates $\mu$ and the combination of them according to \eqref{eq:STL_fragment}. Now let $z^\varphi_t$ be a binary indicator such that $z^\varphi_t=1\Longleftrightarrow (\xs,t)\models\varphi$ and $z^\varphi_t=0\Longleftrightarrow (\xs,t)\not\models\varphi$. The conjunction and disjunction of the predicates, i.e., the Boolean satisfaction of $\varphi$ in \eqref{eq:STL_fragment} can then be encoded as \cite{raman}:

\noindent\textbf{\textit{Conjunction}}
\begin{equation}
\setlength{\jot}{.8pt}
\small
\begin{split}
z^{\varphi}_t&\leq z^{\varphi_{i}}_t,\;\;i=1, \ldots , m\, , \\
z^{\varphi}_t&\geq 1-m+ \sum_{i=1}^{m}z^{\varphi _{i}}_t\, .
\end{split}
\label{eq:z_conjunction}
\end{equation}
\textbf{\textit{Disjunction}}
\begin{equation}
\setlength{\jot}{.8pt}
\small
\begin{split}
z^{\varphi}_t&\geq z^{\varphi_{i}}_t,\;\;i=1, \ldots , m \, ,\\
z^{\varphi}_t&\leq \sum_{i=1}^{m}z^{\varphi _{i}}_t\, .
\end{split}
\label{eq:z_disjunction}
\end{equation}

 Next, we will present the MILP encoding of variables that count the consecutive time instances of satisfaction (i.e., $z^\varphi_t\!=\!1$) and violation (i.e., $z^\varphi_t\!=\!0$), respectively. Such an encoding is proposed in \cite{rodionova2021time} only forward in time to calculate right time robustness. To see signal behaviors such as cumulative temporal relaxation (as introduced in \eqref{eq:temp_relax_metric_a}), we formulate a similar satisfaction/violation counting mechanism both forward and backward in time. Then we follow the recursive definition of the general temporal relaxation metric proposed in \eqref{eq:tau_F}, \eqref{eq:tau_G}, and \eqref{eq:temp_relax_metric} to build MILP constraints representing them.
 
 We will use the variables of $r_t^{1}(\varphi),l_t^{1}(\varphi)\in\mathbb{Z}_{\geq0}$ and $r_t^{0}(\varphi),l_t^{0}(\varphi)\in\mathbb{Z}_{\leq0}$ which denote the number of instants with consecutive satisfaction $(1)$ and violation $(0)$ of $\varphi$ starting from $t$ toward the future $(r)$ and past $(l)$, respectively. The numbers of consecutive past satisfaction and violation instants of $\varphi$ starting from $t$ (for $t'\leq t$) are defined forward in time, respectively, as follows:
    \begin{equation}\label{eq:left_counter}\small
    \begin{split}
        &l^{1}_t(\varphi) = z^\varphi_t\cdot(l^{1}_{t-1}(\varphi)+1),\\
&l^{0}_t(\varphi) = (1-z^\varphi_t)\cdot(l^{0}_{t-1}(\varphi)-1).
    \end{split}
    \end{equation}

Similarly, the numbers of consecutive satisfaction and violation instants of $\varphi$ at future starting from $t$ (for $t'\geq t$) are defined backward in time, respectively, as below
    \begin{equation}\label{eq:right_counter}\small
    \begin{split}
&r^{1}_t(\varphi) = z^\varphi_t\cdot(r^{1}_{t+1}(\varphi)+1),\\
&r^{0}_t(\varphi) = (1-z^\varphi_t)\cdot(r^{0}_{t+1}(\varphi)-1).
    \end{split}
    \end{equation}
    
By construction, $r_t^1(\varphi)$ and $l_t^1(\varphi)$ count the maximum numbers of consecutive instants with $z_{t'}^\varphi = 1$ for $t'\geq t$ and $t'\leq t$, respectively. On the other hand, $r_t^0(\varphi)$ and $l_t^0(\varphi)$ count the maximum numbers of consecutive instants with $z_{t'}^\varphi = 0$ and
multiply these values by -1. 

For the given time intervals $I_F=I_G=[a,b]$, boundary conditions for the calculations in \eqref{eq:left_counter} and \eqref{eq:right_counter} are determined based on the type of temporal operator. For \textit{finally}, the calculations in \eqref{eq:tau_F_MILP} start forward and backward in time with $l^0_{t+a-\ceil*{\gamma_F\cdot\vert I_F\vert}}(\varphi)=0$ and $r^0_{t+b+\ceil*{\gamma_F\cdot\vert I_F\vert}}(\varphi)=0$, respectively, where $\ceil*{\gamma_F\cdot\vert I_F\vert}$ is the exclusive bound for the allowable relaxation. On the other hand, the boundary conditions for the calculations for \textit{globally} in \eqref{eq:tau_G_MILP} is started forward and backward in time with $l^1_{t+a-1}=0$ and $r^1_{t+b+1}=0$, respectively. 

We continue with the MILP encoding of temporal relaxation metrics for \textit{finally} and \textit{globally} operators constructed in \eqref{eq:tau_F} and \eqref{eq:tau_G}, respectively.\\
\textit{\textbf{Finally}}
\begin{equation}\label{eq:tau_F_MILP}
\small
\tau(\xs,F_{[a,b]} \varphi, t)=(z_t^{F_{[a,b]} \varphi}-1)\cdot\frac{\max\big(l^{0}_{t+a}(\varphi), r^{0}_{t+b}(\varphi)\big)}{\ceil*{\gamma_F\cdot\vert I_F\vert}},
\end{equation}

\noindent where the variable $z_t^{F_{[a,b]} \varphi}$ denotes the Boolean satisfaction of the original subtask, i.e., $z_t^{F_{[a,b]} \varphi}=1\Longleftrightarrow(\xs,t)\models F_{[a,b]} \varphi$, and its MILP encoding is similar to \textit{disjunction} in \eqref{eq:z_disjunction} as
\begin{equation}\label{eq:z_finally}
\small
    z_t^{F_{[a,b]}\varphi}=\bigvee_{t'=t+a}^{t+b}z_{t'}^{\varphi}\ .
\end{equation} 

Hence, if the \textit{finally} subtask is already satisfied there is no need to consider relaxation beyond the original interval or demanding more satisfactory time instants. This way, we keep the original requirement of satisfying $\varphi$ at any single instant within $[a,b]$ if possible, or at an instant as close as possible to the original interval otherwise.  

\noindent\textit{\textbf{Globally}}
\begin{equation}\label{eq:tau_G_MILP}
\small
\begin{split}
    &\tau(\xs,G_{[a,b]} \varphi, t)=\\&1-z_t^{G_{[a+\beta,b-\beta]}\varphi}\cdot\Bigg(1-\frac{\vert I_G\vert-l^{1}_{t+\floor*{\frac{a+b}{2}}}(\varphi)- r^{1}_{t+\ceil*{\frac{a+b}{2}}}(\varphi)}{{\gamma_G\!\cdot\!\vert I_G \vert}}\Bigg),
\end{split}
\end{equation}

\noindent where $\beta=\floor*{\gamma_G\!\cdot\!\frac{\vert I_G\vert}{2}}$ and the fraction is the ratio of failed time instants to the allowed relaxation. The variable $z_t^{G_{[a+\beta,b-\beta]} \varphi}$ denotes the Boolean satisfaction of the subtask under maximum allowable relaxation, and its MILP encoding is similar to conjunction in \eqref{eq:z_conjunction} as
\begin{equation}\label{eq:z_globally}
\small
    z_t^{G_{[a+\beta,b-\beta]}\varphi}=\bigwedge_{t'=t+a+\beta}^{t+b-\beta}z_{t'}^{\varphi}\ .
\end{equation} 

Therefore, the \textit{globally} subtask can be satisfied under an allowable relaxation, and otherwise removed. Note that when $\gamma_G=1$, the Boolean satisfaction variable becomes redundant as even a single satisfied time instant is an acceptable relaxation for the $\gamma_G=1$ case. The encoding in \eqref{eq:tau_G_MILP} enforces the satisfaction at the middle and possible relaxations at both ends as required by the definition in \eqref{eq:tau_G}. As $\floor*{\frac{a+b}{2}}$ and $\ceil*{\frac{a+b}{2}}$  denote the same time step when $a+b$ is even, for such cases, we consider the satisfaction at only one of these two instants in the numerator of \eqref{eq:tau_G_MILP} to prevent  double counting.

After encoding the temporal relaxation metrics of $\tau(\xs,F_{[a,b]} \varphi,t)$ and $\tau(\xs,G_{[a,b]} \varphi,t)$, the MILP formulation of $\tau(\xs,\Phi,t)$ for the overall STL specification $\Phi$ can be done according to the recursive definitions in \eqref{eq:temp_relax_metric} and using the recipe in \cite{raman} for quantitative encoding of $\min$/$\max$ functions. 

\SetKwInOut{Input}{input}
\SetKwInOut{Output}{output}
\SetKwInOut{Initialize}{initialize}
\SetKwInOut{Return}{return}
\begin{algorithm}[htb!]\small
\setlength{\jot}{.5pt}
\SetAlgoLined
\Input{Initial state $\xs_{0}$, linear dynamics $f(\xs,\us)$, STL specification $\Phi$, and user tolerance inputs $\gamma_F$ and $\gamma_G$.}
Define constraints of the dynamics along the mission horizon $T$\;
Formulate MILP constraints on $z_t^\varphi$ for each subtask $\varphi$ over $t\in[0,T]$ (i.e., predicates and their conjunction/disjunction)\;
Formulate STL constraints to define $\tau(\xs,\Phi,0)$ built recursively upon the core temporal relaxation metrics $\tau(\xs,F_{[a,b]} \varphi,t)$ in \eqref{eq:tau_F} and $\tau(\xs,G_{[a,b]} \varphi,t)$ in \eqref{eq:tau_G} according to \eqref{eq:temp_relax_metric}\;
Solve \eqref{eq:optimization} in MILP format to extract state and input sequences\; 
\Output{Relaxed STL specification $\Phi^\tau$ with the temporal relaxation amounts on each subtask, state trajectory $\xs$ and input policy $\us^*$ that achieve $\Phi^\tau$.}
\caption{STL Control Synthesis under Minimal Temporal Relaxation using MILP encoding}
\label{alg:1}
\end{algorithm}

{Unlike the time robustness, temporal relaxation metric is not defined throughout the whole mission horizon. In this regard, much less variables are used in MILP encoding compared to time robustness yielding substantially faster solutions in addition to obtaining different signal behaviors. Moreover, counting the violations only for the \textit{finally} and the satisfactions only for the \textit{globally} operator within predetermined bounds ($\overline{I}_F$ and $\overline{I}_G$) keeps the number of optimization variables low and contribute to the computational efficiency.}

\begin{theorem}
For an STL specification $\Phi$ defined using the syntax in \eqref{eq:STL_fragment} with linear predicates, if any subtask $\Phi_j$ in $\Phi=\bigwedge_{i=1}^k \Phi_i$ is feasible alone, then Alg. 1 returns $\tau(\xs,\Phi,t)<1$ implying that the spatial requirement of at least one subtask is satisfied.
\end{theorem}
\vspace{-2mm}
\begin{proof}
 MILP encoding of the temporal relaxation in lines 2-3 of Alg. 1 implies by construction (via Eqs. \eqref{eq:tau_F_MILP} and \eqref{eq:tau_G_MILP}) that the minimizer of the optimization problem solved in line 4 chooses the lowest cumulative temporal relaxation metric among the alternatives. Suppose that $\Phi_j$ is a feasible subtask in $\Phi$. Now we will consider two cases based on the satisfaction status of the $\Phi_j$:
 
 \noindent\textbf{\textit{Case 1 ($\Phi_j$ is completed on time):}} Since no relaxation is needed for $\Phi_j$, even if all other subtasks are completely violated $(\text{i.e.,}\ \tau(\xs,\Phi_i,0)=1,\ \forall i\not=j)$, the resultant cumulative temporal relaxation is bounded as $\tau(\xs,\Phi,0)\leq\frac{\tau(\xs,\Phi_j,0)+k-1}{k}$. It follows from $\tau(\xs,\Phi_j,0)=0$ that $\tau(\xs,\Phi,0)\leq\frac{k-1}{k}<1$. Let us denote $\tau(\xs,\Phi,0)$ in this case as $\tau_{case_1}$ for future reference.
 
 \noindent\textbf{\textit{Case 2 ($\Phi_j$ is temporally relaxed or entirely removed):}}
    Suppose that the feasible subtask $\Phi_j$ is relaxed, and the algorithm still returns $\tau(\xs,\Phi,0)=\frac{\tau(\xs,\Phi_j,0)+\sum_{i\not=j}\tau(\xs,\Phi_i,0)}{k}=1$ meaning no spatial requirement is achieved within the relaxed time intervals. However, for the Alg. 1 to return a temporal relaxation for $\Phi_j$, instead of $\tau_{case_1}$, we need to have $\frac{\tau(\xs,\Phi_j,0)+\sum_{i\not=j}\tau(\xs,\Phi_i,0)}{k}\leq\tau_{case_1}\leq\frac{k-1}{k}$ which is a contradiction and yields $\tau(\xs,\Phi,0)<1$ since $0\leq\tau(\xs,\Phi,0)\leq1$ as presented in Corollary \ref{corollary_1}.
\end{proof}
\vspace{-6mm}
\textcolor{black}{\section{Case Studies} \label{sec:case_studies}}
\begin{table*}[htb!]
\centering
    \caption{Simulation results and optimization parameters for temporal relaxation $\big(\tau(\xs,\Phi_{case},0)\big)$ minimization and standard time robustness $\big(\theta^+(\xs,\Phi_{case},0)$ and $\theta^-(\xs,\Phi_{case},0)\big)$ maximization.}
    \label{table:benchmark}
    \begin{tabular}{|c| c| c c| c c| c|}
     \hline \\ [-1.2em]
     \multirow{2}{*}{Optimized Metric} & \multirow{2}{*}{\# Constraints} & \multicolumn{2}{c}{\# Variables}\vline & \multicolumn{2}{c}{Computation Time $[s]$}\vline & \multirow{2}{*}{\begin{tabular}{@{}c@{}}Temporal \\ Relaxation $\in[0,1]$\end{tabular}}\\
     & &Continuous & Integer & YALMIP & Solver & \\ \hline
     $\tau(\xs,\Phi_{case},0)$  & 5505 & 310 & 1821 & 1.66 & 39.2 & 0.277 \\ \hline
     $\theta^+(\xs,\Phi_{case},0)$  & 8199 & 970 & 1861 & 1.98 & 1118.6 & 0.465 \\ \hline
     $\theta^-(\xs,\Phi_{case},0)$  & 8175 & 964 & 1861 & 2.32 & 632.9 & 0.461 \\ \hline
    \end{tabular}
\end{table*}
\begin{figure*}[htb!]
    \centering
    \includegraphics[trim=120 6 120 15,clip, width=\linewidth]{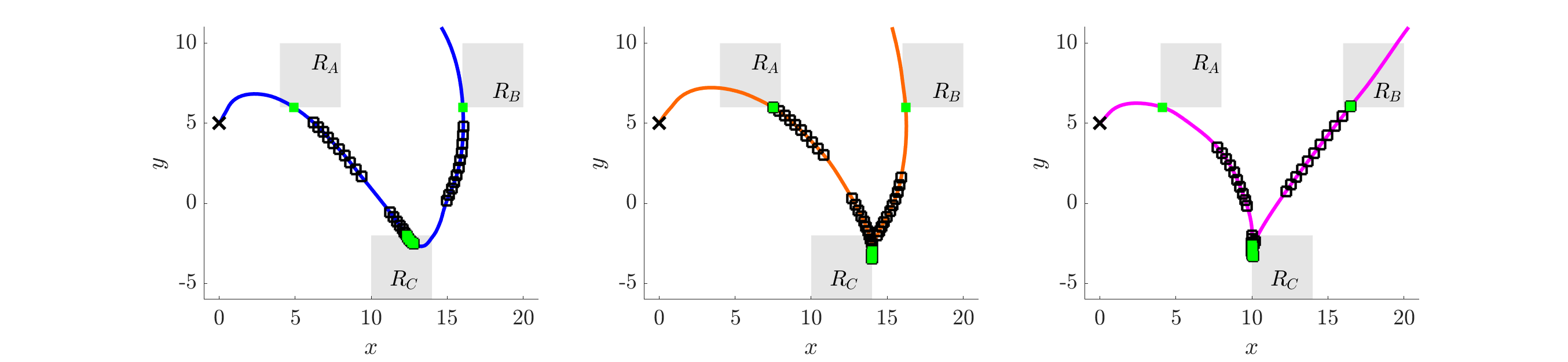}
    \caption{Trajectories of the robot under the proposed minimal temporal relaxation control approach (on the left) together with right and left time robustness maximizing approaches (on the middle and right, respectively). The mission requirements are defined in \eqref{eq:case_STL}. Initial position is marked by $\times$. The original \textit{globally} and \textit{finally} intervals are denoted by discrete black squares where satisfaction instants within or closest to them are represented via green fillings.}
    \label{fig:trajectories}
\end{figure*}
We develop a control synthesis tool that generates trajectories
satisfying the given STL specification under minimal temporal relaxation by solving
the problem in \eqref{eq:optimization}. We formulate the problem 
in YALMIP \cite{yalmip} and solve it using Gurobi \cite{gurobi} in MATLAB R2021a. A laptop computer with 1.8 GHz, Intel Core i5 processor is used to run the simulations with the mission horizon of $T\!=\!120\, s$. We use $\gamma_F\!=\!\gamma_G\!=\!1$ to bound the relaxations over \textit{finally} and \textit{globally} subtasks, respectively. When the subtask intervals are close to $t=0$ or $t=T$, we shift the signals as they are defined on $\mathbb{Z}_{\geq0}$ or extend the mission horizon $T$ accordingly so that the allowable relaxation bounds $\overline{I}_{F}$ and $\overline{I}_{G}$ are captured within $[0,T]$.

To illustrate the benefits of the new temporal relaxation metric, here we address control synthesis for an autonomous robot with discrete-time double-integrator dynamics as
\begin{equation}\label{eq:case_dynamics}
\small
\xs_{t}^+=
   \begin{bmatrix}
      1 & \Delta t & 0 & 0 \\
      0 & 1 & 0 & 0 \\
      0 & 0 & 1 & \Delta t \\
      0 & 0 & 0 & 1
   \end{bmatrix} 
   \xs_t +
   \begin{bmatrix}
   0.5 \Delta t^2 & 0\\
   \Delta t & 0\\
   0 & 0.5 \Delta t^2\\
   0 & \Delta t
   \end{bmatrix} \us_t,
\end{equation}
with the state vector $\xs=[x,v_x,y,v_y]^T$ where $v_x$, $v_y\in\mathbb{R}$ are the velocities in $x$, $y\in\mathbb{R}$ directions, respectively; and the input vector $\us=[u_x,u_y]^T$ where $u_x$,$u_y\in\mathbb{R}$ are the accelerations along the given directions with the limit of $\vert u_x\vert,\!\vert u_y\vert\!\leq\!2.2$. The mission scenario requires the robot to visit $R_A$ and $R_B$ some time within $[32,42]$ and $[77,87]$, respectively, and always stay inside $R_C$ within $[47,67]$. These requirements are expressed by the STL specification:
\begin{equation}\label{eq:case_STL}
\small
\Phi_{case} = F_{[32,42]} R_A \wedge F_{[77,87]} R_B \wedge G_{[47,67]} R_C,    
\end{equation}
where visiting regions can be captured by the conjunction of linear predicates as follows:
\begin{equation}\small
    \begin{split}
        R_A&=x\geq 4\wedge x\leq 8\wedge y\geq 6\wedge y\leq 10,\\
        R_B&=x\geq 16\wedge x\leq 20\wedge y\geq 6\wedge y\leq 10,\\
        R_C&=x\geq 10\wedge x\leq 14\wedge y\geq -6\wedge y\leq -2.\\
    \end{split}
\end{equation}

The STL control synthesis problem under $\Phi_{case}$ is solved by i) minimizing the proposed temporal relaxation metric $\tau(\cdot)$ using Alg. \ref{alg:1} and encoding in Sec. \ref{sec:soln_approach}, ii) maximizing right \big($\theta^+(\cdot)$\big) and left \big($\theta^-(\cdot)$\big) time robustness metrics via the approach in \cite{rodionova2021time} for comparison. Results are given in Table \ref{table:benchmark}, and realized trajectories with the marked satisfactory instances and the original intervals are presented in Fig. \ref{fig:trajectories}. The use of the proposed metric results in smaller temporal relaxation compared to time robustness maximizing trajectories. Furthermore, the less number of constraints and variables in our encoding yields more efficient computation time as well. As a result, the captured signal behaviors via all three approaches in Table \ref{table:benchmark} lead the following realized STL specifications with marked relaxations in the time intervals: 
\textcolor{black}{
\begin{equation}\label{eq:relaxed_case}
\small
\setlength{\jot}{.5pt}
    \begin{split}
        \Phi_{case}^\tau &= F_{[\boldsymbol{28^*},42]} R_A \wedge F_{[77,\boldsymbol{89^*}]} R_B \wedge G_{[\boldsymbol{53^*},67]} R_C,\\
    \Phi_{case}^{\theta^+} &= F_{[32,42]} R_A \wedge F_{[77,\boldsymbol{95^*}]} R_B \wedge G_{[\boldsymbol{59^*},\boldsymbol{65^*}]} R_C,\\     
    \Phi_{case}^{\theta^-} &= F_{[\boldsymbol{21^*},42]} R_A \wedge F_{[77,87]} R_B \wedge G_{[\boldsymbol{50^*},\boldsymbol{62^*}]} R_C.
    \end{split}
\end{equation}
}
As depicted above, using standard time robustness can assess relaxations toward either right or left. Therefore, one of the \textit{finally} subtasks had to be satisfied when maximizing each time robustness metric. Moreover, relaxation on one end of the \textit{globally} interval enables a smaller relaxation on the other end without additional cost, hence causes unnecessary relaxation. The proposed metric, on the other hand, cumulatively assess the relaxations on both ends and among other subtasks, hence satisfy an STL specification with time intervals as close as possible to the original ones.

\begin{figure}[htb!]
    \centering
    \includegraphics[trim=50 6 60 15,clip, width=\linewidth]{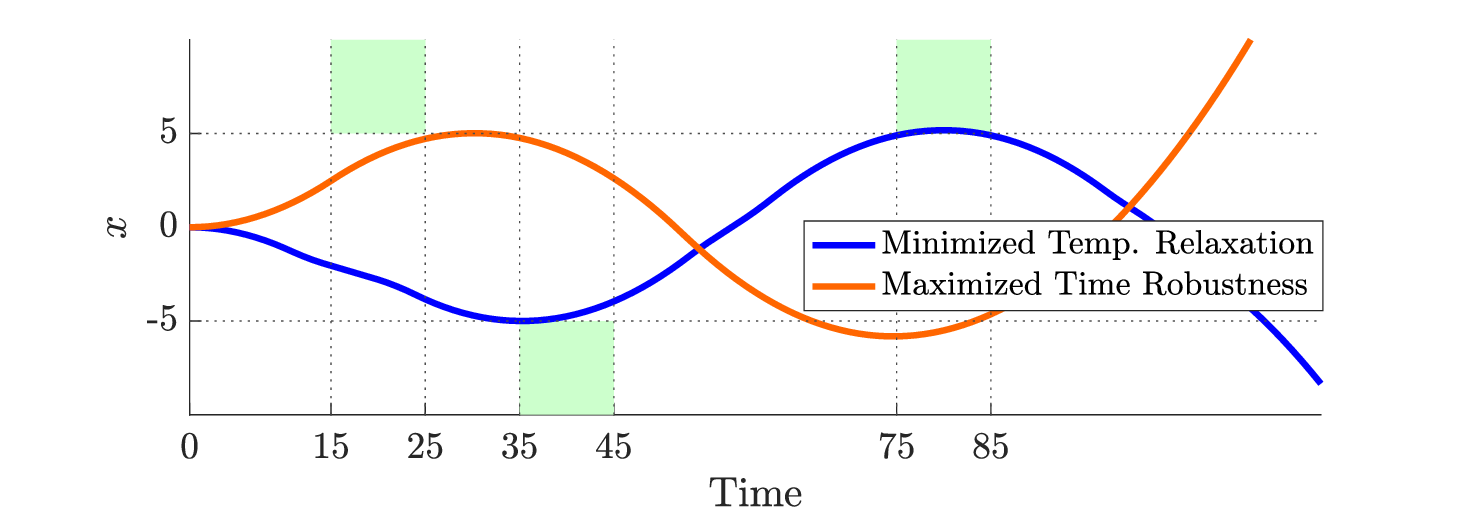}
    \caption{Time history of state $x$. Green regions are desired to be visited via $\Phi_{case'} = G_{[15,25]} x\geq5 \wedge F_{[35,45]} x\leq-5 \wedge F_{[75,85]} x\geq5$.}
    \label{fig:case2}
\end{figure}

\noindent\textbf{\textit{Relaxation via Structural Changes:}} Different from standard STL robustness metrics, the proposed temporal relaxation metric may result in subtask removal to avoid significant cascaded delays in subtask achievement. For instance, consider the specification: $\Phi_{case'}\! =\! G_{[15,25]} x\geq5\! \wedge\! F_{[35,45]} x\leq-5 \! \wedge\! F_{[75,85]} x\geq5$. As the time robustness can be arbitrarily small (and negative), even an utter delay is better than not achieving the subtask. Hence, maximizing it returns the trajectory in Fig.~\ref{fig:case2} which fails all subtasks. In fact, removing the one subtask may lead the satisfaction of others on time. However, the trajectory still tries to bound the violation which yields cascaded delays in the remaining subtasks. On the other hand, minimizing temporal relaxation metric requires the removal of \textit{globally} subtask and enables the completion of others within the original intervals. Hence, Alg. \ref{alg:1} returns relaxed specification of $\Phi^\tau_{case'} = F_{[35,45]} x\leq-5 \wedge F_{[75,85]} x\geq5$ without any relaxation on the time intervals. Moreover, the total solution time for minimizing temporal relaxation is $2.3\, s$ with $\tau(\xs,\Phi_{case'},0)=0.33$, i.e., one-third of the specification is compromised, where the time robustness maximization takes much longer, $59.8\, s$. 
\vspace{-1mm}

\textcolor{black}{\section{Conclusions}
\label{sec:conclusions}}
We introduce a metric that quantifies temporal relaxation of STL specifications. We propose a mixed-integer encoding for temporal relaxation metric and formulate an optimization problem to minimize it. We compare the behavior obtained by minimizing temporal relaxation with the one obtained by maximizing time robustness. We demonstrate that the proposed formulation is computationally efficient and leads to the satisfaction of modified STL specifications by minimally modifying the time intervals under the allowable relaxation limits.
\bibliography{references}
\end{document}